\documentclass{article}

\usepackage{arxiv}

\usepackage[utf8]{inputenc} 
\usepackage[T1]{fontenc}    
\usepackage{url}            
\usepackage{booktabs}       
\usepackage{amsfonts}       
\usepackage{nicefrac}       
\usepackage{microtype}      
\usepackage{lipsum}

\usepackage{relsize,balance,lipsum,bbm,enumerate,times,comment,color,graphicx,setspace,mathdots,mathrsfs,amssymb,latexsym,amsfonts,amsmath,cite,stmaryrd,caption,pgf,accents,mathtools,tabu,enumitem,hhline,array,epstopdf,nicefrac,amsthm,microtype,algorithmic,array,float,bm,url}

\usepackage{graphicx}

\usepackage{mathtools} 

\usepackage[utf8]{inputenc}
\usepackage[english]{babel}

\usepackage{tikz}
\usetikzlibrary{positioning}
\definecolor{mygreen}{RGB}{153,255,153}
\definecolor{myorange}{RGB}{255,178,102}
\definecolor{myred}{RGB}{255,153,153}
\definecolor{myblue}{RGB}{153,204,255}

\title{Robustness of Dynamics in Games: A Contraction Mapping Decomposition Approach}

\author{
    Sina~Arefizadeh\\
    Dept. of Electrical and Computer Engineering\\
    Arizona State University\\
    Tempe, Arizona, USA\\
    \texttt{sarefiza@asu.edu}
\And
    Sadegh~Arefizadeh\\
    Dept. of Electrical and Computer Engineering\\
    Tarbiat Modares University\\
    Tehran, Iran\\
    \texttt{sadegh.arefizadeh@modares.ac.ir}
\And
    S. Rasoul Etesami*\\
    Dept. of Industrial and Systems Engineering\\
    \& Coordinated Science Laboratory\\
    University of Illinois Urbana-Champaign\\
    Champaign, Illinois\\
    \texttt{etesami1@illinois.edu}
\And
    Sadegh Bolouki\\
    Dept. of Mechanical Engineering\\
    Polytechnique Montr\'eal\\
    Montreal, Quebec, Canada\\
    \texttt{sadegh.bolouki@polymtl.ca}
}

\theoremstyle{definition}
\newtheorem{lemma}{Lemma}
\newtheorem{thm}{Theorem}
\newtheorem{prop}{Proposition}

\newtheorem{assum}{Assumption}
\newtheorem{definition}{Definition}
\newtheorem{rem}{Remark}

\begin{document}
\maketitle
\begin{abstract}
A systematic framework for analyzing dynamical attributes of games has not been well-studied except for the special class of potential or near-potential games. In particular, the existing results have shortcomings in determining the asymptotic behavior of a given dynamic in a designated game. Although there is a large body literature on developing convergent dynamics to the Nash equilibrium (NE) of a game, in general, the asymptotic behavior of an underlying dynamic may not be even close to a NE. In this paper, we initiate a new direction towards game dynamics by studying the fundamental properties of the map of dynamics in games. To this aim, we first decompose the map of a given dynamic into contractive and non-contractive parts and then explore the asymptotic behavior of those dynamics using the proximity of such decomposition to contraction mappings. In particular, we analyze the non-contractive behavior for better/best response dynamics in discrete-action space sequential/repeated games and show that the non-contractive part of those dynamics is well-behaved in a certain sense. That allows us to estimate the asymptotic behavior of such dynamics using a neighborhood around the fixed point of their contractive part proxy. Finally, we demonstrate the practicality of our framework via an example from duopoly Cournot games.
\end{abstract}

\section{INTRODUCTION}\label{intro}

Noncooperative game theory is one of the key subjects in multi-agent decision systems, in which multiple agents/players strategically take actions to maximize their own payoff functions. Among various classes of noncooperative games, \emph{potential games} constitute a special class, in which there exists a common potential function that is aligned with the incentives of all the individuals' payoff functions \cite{C1}. Potential games cover a wide range of applications such as traffic routing in transportation systems \cite{C2}, power allocation in cognitive radio \cite{C3}, and task scheduling in robotics \cite{C4}. Moreover, for potential games, many distributed learning dynamics, including the best response or fictitious play types, are known to converge to a Nash equilibrium (NE)--a joint action profile in which no player can benefit by deviating unilaterally. In fact, a viable approach for modeling rational players in a strategic environment is the convergence of dynamics resulting from players' interactions to a NE. That raises a natural question on whether a NE can emerge as a long-run behavior of the players' decisions. While the answer to such a question for some special class of games and dynamics is known, the current literature lacks a systematic analysis for determining the asymptotic behavior of arbitrary dynamics in general games \cite{C5,C6}. Therefore, our main objective in this work is to provide an analytic framework to study the asymptotic behavior of dynamics in games using results from contraction theory.

\subsection{Related Work}

In a game setting, depending on how players take actions, various dynamics can emerge. Fictitious play (FP) dynamics is one of the most well-known learning dynamics in games, in which each player chooses her best action to the realized action history of other players and the state of the underlying environment. It is known that the FP dynamics converge to a NE for certain classes of games such as two-player zero-sum games if the players have perfect information on others' play and the underlying environment. Moreover, it is known that best response (BR) dynamics and FP dynamics converge to NE of {\it potential games} \cite{C1,C5,C7,C8}. In addition, in potential games with uncertain environments and fixed communication networks, it is known that FP dynamics converge to a NE \cite{C4}. The paper \cite{C10} shows the convergence of beliefs to a NE of potential games when the information on players' actions is scattered throughout the network. It also provides an extension of \cite{C4} to address the problem in the presence of time-varying communication connectivity. The work \cite{C14}  further extends the results of \cite{C9} and \cite{C10} on FP dynamics to the variant of fictitious introduced in \cite{C10} for near-potential games. Authors of \cite{C9} show that {\it near-potential games} share some similar characteristics with close-potential games.\footnote{A near-potential game can be viewed as a deterministic or stochastic perturbation of a potential game. A close potential game refers to a potential game that is close to the main game in the sense of a certain metric. We refer to Section \ref{Se_Nom} for a formal definition.}  For instance, if FP dynamics converge to a NE in a potential game, they may converge to some neighborhood of that NE in a corresponding near-potential game. The work \cite{C9} also investigates other dynamics such as BR dynamics and Logit dynamics in sequential games. Other update rules in repeated games, such as the multiplicative weight update rule, have been proved to converge to a Nash equilibrium under certain assumptions \cite{C32,C30,C31}. However, the behavior of such update rules under perturbed games has not been analyzed before. In another line of works \cite{C33,C34,C35}, information is used to change the equilibrium outcome, and the convergence of the best response dynamics is addressed under such environments.

Our work is also related to \cite{C11,C12,C13}, in which different notions of perturbed games using deterministic, stochastic, or set-based uncertainties were introduced and studied. For instance, the works \cite{C12,C13} introduce the concept of payoff shifters, which is equivalent to perturbing the utility functions of a game by some additive terms. This notion is then utilized to study the equilibrium in markets in which players' strategies are based on their perceived utilities of some actual market. On the other hand, the work \cite{C13} defines a convex set of games consistent with a set of observed perturbed games and uses convex optimization techniques to find the closest consistent game to the set of observed perturbed games.

An alternative approach for modeling uncertainty in games is through utilizing the notion of approximate equilibrium, which can be used to address perturbation in utility functions \cite{C26}. For instance, instead of computing an exact NE, one can consider finding an $\epsilon$-NE, in which a player's deviation cannot improve its payoff by more than $\epsilon$. Additionally, it was shown in \cite{C21} that computing NE and $\epsilon$-NE for $\epsilon=O(1/\mbox{poly}(n))$ is PPAD-complete, and given small values of $\epsilon$, one cannot hope polynomial-time algorithms for computing an $\epsilon$-NE \cite{C24}. That has motivated some of the prior works \cite{C22,C23} to devise polynomial-time algorithms for computing an $\epsilon$-NE when $\epsilon$ is sufficiently large. We should note that perturbation in games is not always modeled in utility functions. For instance, in the trembling hand perfect equilibrium, the uncertainty in players' preferences is often modeled using a small probability that players may take suboptimal actions. We refer to \cite{C25} for analyzing such dynamics with perturbations in the players' actions.

\subsection{Contributions}


In this work, by decomposing the mapping of the game dynamics to contractive and noncontractive parts, we investigate a new direction to study the mapping of the dynamical systems in games. Our approach provides a paradigm for analyzing the limiting behavior of learning dynamics in games. More precisely, we view a game by an associated mapping that takes as its input a point in the players' action space and maps it to another point in the same space. That allows us to study the dynamic behavior of a game through iterative compositions of its associated mapping. Inspired by some ideas from near-potential games, we will introduce a notion of near-contractive mapping for the map of a generic game and leverage it to study the asymptotic behavior of game dynamics.  To this aim, we first find a contraction mapping close to the mapping of a given game so that we can decompose the game into the contractive and non-contractive parts. Having some information on the given game, we show that the game outcome will eventually be trapped in a neighborhood of the fixed point of the contractive mapping. As a result, we can approximate the asymptotic behavior of the game dynamics using the fixed point of that contractive mapping.

 The idea of decomposing a game into components is not specific to this work. For instance, \cite{C20} provides a game decomposition into non-strategic, potential, and harmonic components and leverages that decomposition to infer some information about the equilibrium points. However, such a decomposition may not be very helpful in a dynamic setting with unknown limiting behavior. The reason is that players' dynamic behavior depends not only on their utility functions but also on their action update rule. Therefore, in this work, we are interested in decomposing the mapping of the dynamics in game rather than the game itself, which can be different from the decomposition given in \cite{C20}. The latter study focuses on presenting the space of games and breaking it down into linear subspaces. This decomposition is carried out with respect to the utility functions of the appropriate game and is independent of any dynamics. In this work, we instead offer an analysis of the infinite time composition of the mapping of a given dynamic on itself by decomposing the map of dynamics into contractive and non-contractive parts. The main application of the analysis given in this current paper is about robustness in games, where robustness refers to the sensitivity of the limit behavior of the game dynamics to the changes in the mapping of the corresponding dynamics due to perturbation in utilities. Despite earlier works that focus on the equilibrium of the perturbed games in a static setting \cite{C12,C13}, we utilize the concept of payoff shifters to model the perturbations in a dynamic setting. In other words, we consider a scenario in which players respond to their opponents using perturbed utilities and based on some specific update rule. Here, the perturbed utilities may belong to some specific set with certain bounds on the Lipschitz constant of the perturbation functions.

Our work is mostly related to the paper \cite{C9} that studies the asymptotic behavior of near-potential games for specific sequential dynamics such as Logit dynamics, BR dynamics, and FP dynamics. In the present work, however, we identify a tighter asymptotic behavior for BR dynamics in sequential near-potential games with discrete-action spaces. In addition, we discuss the asymptotic behavior of BR dynamics in repeated games, which was not addressed in \cite{C9}. Lastly, we illustrate an implication of our framework on the class of duopoly Cournot games.  Our contributions can be summarized as follows:
\begin{itemize}
    \item Introduction of a new framework to analyze the limit behavior of dynamics in games (Theorem \ref{Th:1}) with respect to closely related contractive dynamic. This framework 
 and Theorem \ref{Th:1} in particular has a wide spectrum of applications. In fact, as long as we can guarantee that the limit behavior is bounded or the actions space is compact this theorem and the framework work.
    \item Analyzing the asymptotic behavior of the best response and better response dynamics in sequential discrete action space games using the proposed framework (Theorems \ref{Th:1} and \ref{Th:2}) with respect to closely related contractive dynamic.
    \item Analyzing the asymptotic behavior of the best response dynamics in repeated games under continuous and discrete action spaces (Theorems \ref{Th:1} and \ref{th:4}) with respect to closely related contractive dynamic.
    \item Analyzing the asymptotic behavior of the aforementioned dynamics in games when players only have an estimate of others' actions, but their estimations will become accurate over time with respect to closely related contractive dynamic. This extends our proposed framework to the settings where players respond to estimations of other players' decisions.
\end{itemize}

\subsection{Organization}
In Section \ref{Se_Nom}, we introduce some basic notions and terminology. In Section \ref{sec:Pro_St}, we formally introduce the problem. In Section \ref{sec:Main_results}, we provide our main results on the asymptotic behavior of the sequential update dynamics in discrete-action space, assuming that the path of play is eventually trapped in a bounded region in the action space. We extend those results to analyze the map of simultaneous update dynamics in repeated games and generalize our findings to situations where agents' estimates of others' actions are inaccurate and gradually improve over time. We provide an application of our framework on a special class of duopoly Cournot games  in addition to simulation of a special setting of this game in Section \ref{Example_and_Simulation}. We conclude the paper in Section \ref{Conclusion} and relegate omitted proofs to Appendix I.


\section{Notions and Terminology} \label{Se_Nom}

Let us consider a noncooperative game with a set $\mathcal{M}= \{1,\ldots,N\}$ of players that are connected by a fixed undirected underlying graph $G(\mathcal{M},E)$, with node set $\mathcal{M}$ and edge set $E$. We denote the finite action set of player $i\in \mathcal{M}$ by $\mathcal{A}_i \subset \mathbb{R}$, and the joint action set of all the players by $\mathcal{A}=\prod_{i=1}^N\mathcal{A}_i$. Moreover, we assume that the perturbed utility function of player $i\in \mathcal{M}$, denoted by $u_i^M:\mathcal{A}\rightarrow \mathbb{R}$, can be expressed as 
\begin{equation}\label{eq_1}
u_i^M=u_i+\Delta u_i,
\end{equation}where $u_i: \mathcal{A}\rightarrow \mathbb{R}$, $i\in \mathcal{M}$ are players' \emph{nominal} utility functions, which are assumed to be independent for different players $i$. The term $\Delta u_i$ denotes the perturbation function to the player $i$'s nominal utility function. Such a setting may arise due to adversarial behavior in multi-agent networks where each agent affects other agents' behavior.

We adopt the notation $a_{-j}=(a_1,\ldots,a_{j-1},a_{j+1},\ldots,a_N)$. According to \eqref{eq_1}, the perturbed utility of player $i$, $u_i^M$, is a summation of her own nominal 
utility function and a perturbation function. Therefore, we can view the game with utility functions $u_i^M$, denoted by $\Gamma^M=(\mathcal{M},\{\mathcal{A}_i,u_i^M\}_{i\in\mathcal{M}})$, as a perturbed version of the so-called \emph{original game} $\Gamma=(\mathcal{M},\{\mathcal{A}_i,u_i\}_{i\in\mathcal{M}})$ with nominal utility functions $u_i$. Here, the superscript $M$ stands for the \emph{perturbed} game. 

The following definition provides the concept of convex hull of a given set in a Euclidean space.
\begin{definition}{\normalfont(Convex hull, ~\cite{C36})} The convex hull of a given set $X$ in a Euclidean space is the minimal convex set containing $X$. We denote the convex hull of the set $X\subseteq \mathbb{R}^N$ by $Conv(X)$.
\end{definition}

As we mentioned earlier, potential games form a well-studied class of games for which the asymptotic behavior of various dynamics is known. This class of games plays a key role in the study of near-potential games, which are in the proximity of potential games. The following definition introduces the class of potential games \cite{C1}.  
\begin{definition}{\normalfont(Potential Game, \cite[Definition~1]{C14})}\label{def:1}
A game with utility functions $u_i$ and joint action set $\mathcal{A}=\prod_{i=1}^N\mathcal{A}_i$ is called an (exact) potential game if there exists a global function $\phi:\mathcal{A}\rightarrow \mathbb{R}$ such that for any player $i$, and two actions $a_i, a'_i \in \mathcal{A}_i$, and any $a_{-i}\in \mathcal{A}_{-i}$ (where $\mathcal{A}_{-i}$ is the Cartesian product of  $\mathcal{A}_j$ for $j\in\mathcal{M}\setminus\{i\} $), we have
\begin{align}\label{pot}
u_i(a_i,a_{-i})-u_i(a'_i,a_{-i})=\phi(a_i,a_{-i})-\phi(a'_i,a_{-i}).
\end{align}
\end{definition}

Next, let us assume $Conv(\mathcal{A})\subset D$, where $D$ is some bounded disk in $\mathbb{R}^N$. For any mapping $F$, let $F^n$ denote $n$-times composition of mapping $F$ by itself. Given an update rule dynamics for players' actions, let us denote the game mapping under that update rule for the perturbed game $\Gamma^M$ and the original game $\Gamma$ by $K:D\rightarrow D$ and $Z:D\rightarrow D$, respectively. Then, one can use the following iterative maps to describe the dynamic behavior of the games $\Gamma^M$ and $\Gamma$ as

\begin{equation}\label{Map_K}
x^{t}=K^t(x^0),
\end{equation}
\begin{equation}\label{Map_Z}
y^{t}=Z^t(y^0),
\end{equation}
where $x^0$ and $y^0$ are some initial action profiles in $\mathcal{A}$. Sequences  $\{x^t\}_{t=1}^\infty$ and $\{y^t\}_{t=1}^\infty$ show the trajectory of systems (\ref{Map_K}) and (\ref{Map_Z}) in $D$, respectively. For instance, one can represent the BR update rule for simultaneous update in repeated games in the perturbed and original games by the mappings $K=(K_1,K_2,...,K_N)$ and $Z=(Z_1,Z_2,...,Z_N)$, where 
\begin{align}
&K_i(a)\in{\arg\max}_{a_{i}\in \mathcal{A}_i} u^M_i(a_{i},a_{-i})\cr
&Z_i(a)\in{ \arg\max}_{a_{i}\in \mathcal{A}_i} u_i(a_{i},a_{-i}).
\end{align}

In the remainder of this section, we provide definitions of NE and game distance, which will be useful for our subsequent analysis. To address strategy profiles that are close to a NE, we use the concept of $\epsilon$-Nash equilibrium ($\epsilon$-NE). 
\begin{definition}{\normalfont($\epsilon$-Nash equilibrium, \cite[Definition~4]{C14})} \label{epsilon-eq}
For an arbitrary game $G=(\mathcal{M}, \{\mathcal{A}_i,u_i\}_{i\in\mathcal{M}})$, an action profile $p=(p_1,..., p_N)$ is called an $\epsilon$-NE ($\epsilon \geq 0$) if

\begin{equation}
u_i(p'_i;p_{-i})- u_i (p_i ; p_{-i})\leq \epsilon
\end{equation}

for every $p'_i \in \mathcal{A}_i$ and any player $i \in \mathcal{M}$. We denote the set of $\epsilon$-NE by $\mathcal{X}_\epsilon$. Note
that a NE is an $\epsilon$-NE when $\epsilon=0$.
\end{definition}

\begin{definition}{\normalfont(Maximum Pairwise Difference-MPD, \cite[Definition~2]{C14})} \label{def_mpd}
Let $G=(\mathcal{M},\{\mathcal{A}_i,u_i\}_{i\in\mathcal{M}})$ and  $\hat{G}=(\mathcal{M},\{\mathcal{A}_i,\hat{u}_i\}_{i\in\mathcal{M}})$ be two  games sharing the same set of players and action sets. For an action profile $a=(a_i,a_{-i}) \in \mathcal{A}$, let $d_{(a'_i,a)}^{G}:=u_i(a'_i,a_{-i})-u_i(a_i,a_{-i})$ be the difference in player $i$'s utility by changing its action to $a'_i \in \mathcal{A}_i$ (similarly we can define $d_{(a'_i,a)}^{\hat{G}}$). Then, the maximum pairwise difference between the games $G$ and $\hat{G}$ is defined by
\begin{equation}\label{eq_mpd}
    d(G, \hat{G}):=  \underset{i \in \mathcal{M}, \, a'_i \in \mathcal{A}_i,\, a  \in \mathcal{A} }{\max} |d_{(a'_i,a)}^{G}- d_{(a_i',a)}^{\hat{G}}|.
\end{equation}
\end{definition}





Finally, we adopt the following notations throughout the paper. Unless stated otherwise, we let $N_r(\zeta)$ be an open ball of radius $r$ around $\zeta\in \mathbb{R}^N$ with respect to the Euclidean distance. We denote the closure of $N_r(\zeta)$ by $N_r[\zeta]$. Moreover, we use $\overline{\{\eta^n\}}$ to denote the set of limit points of a sequence $\{\eta^n\}_{n=1}^\infty$. Additionally, we denote power $p$ of a scalar $\eta$ by $\eta^{(p)}$.

\section{Problem Formulation}\label{sec:Pro_St}
We now formally define the problem of interest regarding the asymptotic behavior of dynamics in games. Let us consider a game $\Gamma^M=(\mathcal{M},\{\mathcal{A}_i,u^M_i\}_{i\in\mathcal{M}})$ with players set $\mathcal{M}$, action sets $\mathcal{A}_i, i\in\mathcal{M}$, and perturbed utility functions $u^M_i$ as defined in Section~\ref{Se_Nom}. Given the game $\Gamma^M$, one can define a dynamic decision-making process in which players iteratively update their actions according to some update rule and observations made on others' latest actions. Consequently, a \emph{path of play} is formed by following such an iterative decision-making process, which can be viewed as the trajectory of the dynamics (\ref{Map_K}). Similarly, one can define a path of play for  $\Gamma=(\mathcal{M},\{\mathcal{A}_i,u_i\}_{i\in\mathcal{M}})$ with corresponding dynamics (\ref{Map_Z}).

Our main objective in this work is to see whether the dynamics of the perturbed games (\ref{Map_K}) can be formulated as perturbation of those in original game (\ref{Map_Z}). In particular, we wish to derive sufficient conditions that can relate the convergence behaviour of the perturbed dynamics (\ref{Map_K}) to that of the original dynamics (\ref{Map_Z}).

\section{Main Results for Dynamics in Games}\label{sec:Main_results}
In this section, we first focus on the mapping associated with a general game and investigate its limiting behavior, assuming that the path of play eventually converges to some bounded domain. We then relax this assumption by deriving sufficient conditions under which the path of play for the BR dynamics in sequential discrete-action games (Section \ref{disc_act_sp_seq_ga}), and the BR dynamics in repeated discrete-action or continuous-action (Section \ref{Sub_sec_rep_game}) games always converges to some bounded set.




\begin{lemma}\label{lem1}
Let $\{\delta^i\}_{i=1}^{\infty}$ and $\{p^i\}_{i=1}^{\infty}$ be two sequences of non negative real numbers, and $\{\delta^i\}_{i=1}^{\infty}$ be a bounded sequence. Given an arbitrary $p^0$, and $L_F<1$, assume that
\begin{equation}
    p^n\leq L_F^{(n)} p^0+\sum_{k=1}^{n-1}\delta^k L_F^{(n-k)},\ \forall n\ge 1.
\end{equation}
Then, we have $\limsup_{n} p^n \leq \frac{1}{1-L_F}\limsup_{n}\delta^n$.
\end{lemma}
\begin{proof}
Since $\{\delta^{n}\}_{n=1}^{\infty}$ is a bounded sequence, we have
\begin{equation}
\limsup_{n} \delta^n=\max \overline{\{\delta^n\}}:=\alpha.
\end{equation}
According to Lemma \ref{lem0} in Appendix I, for any $\epsilon>0$, there are infinitely many $n$ such that $\alpha-\epsilon<\delta^n$, and only a finite number of $n$ such that $\alpha+\epsilon<\delta^n$. Let us denote the largest index among those finitely many indices by $N_\epsilon$. Therefore,  for some $n\leq N_\epsilon, \alpha+\epsilon<\delta^n$,  and $\delta^n\leq \alpha+\epsilon, \forall n > N_\epsilon$. Let us define $\delta_\epsilon=\max_{n\leq N_\epsilon} \delta^n$, and consider the sequence $\{\Delta^i\}_{i=1}^{\infty}$ defined by 
\begin{align}
\Delta^i=\begin{cases}\delta_\epsilon & \mbox{if} \ i\leq N_\epsilon\\
\alpha+\epsilon & \mbox{otherwise}.
\end{cases} 
\end{align}
This sequence is decreasing and each of its elements is greater than the corresponding element of $\{\delta^i\}_{i=1}^{\infty}$. In particular, $\lim_{n\rightarrow \infty} \Delta^n=\alpha+\epsilon$. Now, we can write 
\begin{align}\label{eq:1}
p^n\leq L_F^{(n)} p^0+\sum_{k=1}^{n-1}\delta^k L_F^{(n-k)}\leq L_F^{(n)} p^0+\sum_{k=1}^{n-1}\Delta^k L_F^{(n-k)}.
\end{align} Since $\Delta^{n-1}\leq\cdot\cdot\cdot\leq\Delta^{1}$ and  $L_F^{(1)}\geq\cdot\cdot\cdot\geq L_F^{(n-1)}$, we can use Chebyshev's sum inequality to conclude that 
\begin{equation}
\Delta^{n-1}L_F^{(1)}+\cdot\cdot\cdot+\Delta^{1}L_F^{(n-1)}\leq \left(\frac{1}{n-1}\sum_{i=1}^{n-1}\Delta^i\right)\sum_{i=1}^{n-1}L_F^{(i)}.
\end{equation}
Substituting this relation into \eqref{eq:1}, we have $p^n\leq L_F\Delta_{\mbox{ave}}/(1-L_F)$ as $n\rightarrow \infty$, where $\Delta_{\mbox{ave}}$ denotes the average of all elements of the sequence $\{\Delta^i\}_{i=1}^{\infty}$. Since this sequence converges to $\limsup_{i}\delta^i+\epsilon$, the average of its elements converges to the same value. Hence, as $n\to \infty$, the sequence  $\{p^i\}_{i=1}^{\infty}$ is eventually less than  $L_F(\limsup_{i}\delta^i+\epsilon)/(1-L_F), \forall \epsilon>0$. Therefore, we get 
\begin{equation}
\limsup_{n} p^n - \frac{L_F\limsup_{n}\delta^n}{1-L_F}\leq \frac{L_F\epsilon}{1-L_F}, \ \forall \epsilon>0.
\end{equation}
This completes the proof by letting $\epsilon\to 0$ and noting that $L_F<1$.
\end{proof}

\begin{definition}{\normalfont(Contractive Mapping)}\label{Cont_map}
A map $C:\mathbb{R}^N \rightarrow \mathbb{R}^N$ is called contractive if 
\begin{equation}
\|C(x)-C(y)\|\leq L_C\|x-y\|,
\end{equation}
for some $L_C<1$ and all $x,y\in \mathbb{R}^N$. Any contractive map $C$ admits a fixed point $x^*$, such that $C(x^*)=x^*$.  
\end{definition}

Our next goal is to estimate the mapping of a dynamic associated with a given game using contractive mappings. The following lemma whose proof is given in Appendix I shows the existence of near-contractive mappings for a given map.

\begin{lemma}\label{lem2}
Let $Z:D\rightarrow D\subseteq \mathbb{R}^N$ be a mapping. Given any bounded domain $D'\subset D$, there exist infinitely many bounded contractive mappings $C:D\rightarrow D$, such that $\|Z(x)-C(x)\|\leq \delta, \forall x\in D'$ for some $\delta \geq 0$.
\end{lemma}

As we will see later, the term $\delta$ in lemma \eqref{lem2} plays a crucial role in achieving desirable results regarding the limit behavior of dynamics. In particular, for smaller values of $\delta$ stronger convergence results can be obtained. The following remark highlights some facts regarding the parameter $\delta$.
\begin{rem}
The term $\delta$ appearing in the statement of Lemma \ref{lem2} can vary from the diameter of $D$ (if $\alpha=0$ in the proof of the lemma) to $0$ (if $\alpha=1$). While it is difficult to claim anything more for generic games, for special types of games, the calculations of Lemma \ref{lem2} might be simplified to obtain $\delta$ values that are significantly smaller than the diameter of $D$. For instance, in an extreme case where the mapping of the game dynamic belongs to the set of contractive mappings, this term becomes zero. Moreover, there could be other cases where one can obtain a closed-form for $\delta$ through solving the proposed mathematical programming in the proof of Lemma \ref{lem2}. This will be useful in the analysis of our proposed framework.
\end{rem}

Next, we state our main result that estimates the behavior of a given mapping using a contractive mapping.



\begin{thm}\label{Th:1}
Let $Z,K:D\rightarrow D\subseteq \mathbb{R}^N$ be two mappings such that for some $\delta_1\geq0$, $\|K(x)-Z(x)\|\leq \delta_1, \forall x\in D$. Then, the following hold:
\begin{enumerate}
    \item For a given $x^0\in D$, if the sequences $K^n(x^0)$ and $Z^n(x^0)$ are eventually trapped in a bounded domain $D'\subset D$, then there exists a contractive mapping $C:D\rightarrow D$ with Lipschitz constant $L_C$, such that for some $\delta_2>0$, $\|Z(x)-C(x)\| \leq \delta_2$, $\forall x\in D'$ . Moreover, as $n\rightarrow \infty$, the sequences $Z^n(x^0)$ and $K^n(x^0)$ converge to $N_{\delta_2/(1-L_C)}(x^*)$ and $N_{(\delta_2+\delta_1)/(1-L_C)}(x^*)$, respectively, where $x^*$ is the fixed point of $C$. In addition, if $Z^n(x^0)$ converges to $\Tilde{x}$, then $K^n(x^0)$ converges to the neighborhood set $N_{(2\delta_2+\delta_1)/(1-L_C)}(\Tilde{x})$.\label{Th:1_1}
    \item Assume that $Z^m:D\rightarrow D$ is a contractive map for some $m\in \mathbb{N}$, with Lipschitz constant $L_C$ and a fixed point $\Tilde{x}$. Moreover, assume that $K$ is a map with Lipschitz constant $L_K$, and $Z^{mn}(x^0)$ converges to $\Tilde{x}$ as $n\rightarrow \infty$. Then for any integer $0<r<m$, the map $K^{mn+r}$ approaches the set $N_{R_r}[\Tilde{x}]$ as $n\rightarrow \infty$, where $R_r=L_{K^r}(1-{L_K}^{(m)})\delta_1/\big((1-L_C)(1-{L_K})\big)+\|K^r(\Tilde{x})-\Tilde{x}\|$, and $L_{K^r}$ is the Lipschitz constant for $K^r$.\label{Th:1_2}
\end{enumerate}
\end{thm}
\begin{proof}
1) The existence of the contractive mapping $C$ in the statement of the theorem follows from Lemma \ref{lem2}, whose proof gives a constructive method for finding $C$. Let $e^0=Z(x^0)-C(x^0)$, and iteratively define $e^{n}=Z^{n+1}(x^0)-C^{n+1}(x^0)$ and $b^n=\|Z(Z^n(x^0))-C(Z^n(x^0))\|$. Then, we have
\begin{equation}\label{eq:3}
\|e^{n}\|\leq b^n+L_C\|e^{n-1}\|. 
\end{equation}
Expanding the recursive formula in (\ref{eq:3}) and writing it in the form of Lemma \ref{lem1}, and noting that $\limsup_{i}b^i=\delta_2$, we conclude that $\|e^{n}\|$ eventually becomes less than $\delta_2/(1-L_C)$. Let us define $E^{n}=K^{n+1}(x^0)-C^{n+1}(x^0)$ and $r^n=\|Z(K^n(x^0))-C(K^n(x^0))\|$. Using triangle inequality, we can write \begin{align}
\|K^{n+1}(x^0)-C^{n+1}(x^0)\|&\leq \|K\big(K^n(x^0)\big)-Z\big(K^n(x^0)\big)\|\cr 
&+\|Z\big(K^n(x^0)\big)-C\big(K^n(x^0)\big)\|\cr 
&+\|C\big(K^n(x^0)\big)-C\big(C^n(x^0)\big)\|,
\end{align}
which implies,
\begin{equation}\label{eq:4}
\|E^{n}\|\leq \delta_1+r^n+L_C\|E^{n-1}\|. 
\end{equation}
Similarly, using Lemma \ref{lem1} and  noting that  $\limsup_{i}r^i$ is $\delta_2$, we get that $\|E^{n}\|$ eventually becomes less than $(\delta_1+\delta_2)/(1-L_C)$. Therefore, as $n\rightarrow \infty$, the sequences $Z^n(x^0)$ and $K^n(x^0)$ converge to the sets $N_{\delta_2/(1-L_C)}(x^*)$ and $N_{(\delta_2+\delta_1)/(1-L_C)}(x^*)$, respectively, where $x^*$ is a fixed point of $C$. If in addition we assume that $Z^n(x^0)$ converges to $\Tilde{x}$, then by the geometry of the problem, it should be easy to see that the sequence $K^n(x^0)$ must converge to the set $N_{(2\delta_2+\delta_1)(1-L_C)}(\Tilde{x})$.

2) We prove this part in three steps. Firstly, using the Lipschitz continuity of $K$, we can write 
\begin{align}
\|K\big(K^{m-1}(x)\big)\!-\!Z\big(Z^{m-1}(x)\big)\|\leq L_K\|K^{m-1}(x)\!-\!Z^{m-1}(x)\|&\cr+\delta_1.&
\end{align}
By expanding this inequality recursively, we obtain
\begin{equation}
\|K\big(K^{m-1}(x)\big)-Z\big(Z^{m-1}(x)\big)\|\leq \frac{1-{L_K}^{(m)}}{1-{L_K}}\delta_1.
\end{equation}

Secondly, using the triangle inequality, we can write 
    \begin{align}
\|K^{m(n+1)}(x)-C^{n+1}(x)\|&\leq \|K^{m(n+1)}(x)-Z^{m}(K^{mn}(x))\|\cr 
&+\|Z^{m}\big(K^{mn}(x)\big)-C\big(K^{mn}(x)\big)\|\cr 
&+\|C\big(K^{mn}(x)\big)-C\big(C^n(x)\big)\|.
\end{align}
Thus, using the first step together with Lemma \ref{lem1}, we obtain
\begin{align}\label{eq:k-c-l-m}
\|K^{m(n+1)}(x)-C^{n+1}(x)\|\leq \frac{1-{L_K}^{(m)}}{(1-L_C)(1-{L_K})}\delta_1. 
\end{align}

Finally, we can write 
    \begin{align}
&\|K^{m(n+1)+r}(x)-C^{(n+1)+1}(x)\| \leq \cr &\|K^r\big(C^{n+1}\big)(x)-C\big(C^{n+1}(x)\big)\| +\cr
&\|K^r\big(K^{m(n+1)}\big)(x)-K^r\big(C^{n+1}\big)(x)\|.
\end{align}
    Using the Lipschitz continuity of $K^r$ and \eqref{eq:k-c-l-m}, the right hand side of the above inequality as $n\rightarrow \infty$ becomes
    \begin{equation}
    R_r=L_{K^r}\frac{1-{L_K}^{(m)}}{(1-L_C)(1-{L_K})}\delta_1+\|K^r(\Tilde{x})-\Tilde{x}\|,
    \end{equation}
which completes the proof.
\end{proof}

In fact, let $m$ be an integer for which $Z^m$ is a contractive map. Using Theorem \ref{Th:1} repeatedly for different values of $0<r<m$, one can see that the perturbed map $K$ will converge to a neighborhood of the limit point of $Z^m$ with a radius of at most $\max_{0<r<m} R_r$. Finally, we note that as $D'$ is a subset of the bounded set $D$, in the worst-case scenario, one can use Theorem \ref{Th:1} by taking $D'=D$. However, for certain dynamics, a smaller set for $D'$ is known. We shall discuss this point in the subsequent sections.

\subsection{Better/Best Response Dynamics in Discrete-Action Sequential Games}\label{disc_act_sp_seq_ga}
The contraction analysis in Theorem \ref{Th:1} relies on the existence of a bounded set $D'$. This section focuses on  better/best response dynamics in discrete-action sequential games, where the dynamics improve a single player's utility at each step. We then derive some sufficient conditions to guarantee the existence of the suitable set $D'\subset D$ for such a class of dynamics,  giving us better results compared to the case when $D'=D$.

\begin{definition}{\normalfont (Better and Best Response Dynamics)}\label{Bet_Bes_Dyn}
Given a game $G=(\mathcal{M},\{\mathcal{A}_i,u_i\}_{i\in\mathcal{M}})$, an action update process $\{a^t\}_{t\ge 1}$ for the players is called the \emph{best response} (BR) dynamics if  any arbitrary player $i$ is chosen infinitely many times and updates her action to $a^t_i\in \arg\max_{a_i\in \mathcal{A}_i}u_i(a_i,a^{t-1}_{-i})$. The dynamics are called \emph{better response} if a weaker condition, that is $u_i(a^t_i,a^{t-1}_{-i})>u_i(a^{t-1})$, is satisfied.
\end{definition}

In practice, players who are not eligible to improve their utilities will not change the action profile. Therefore, at each time step $t=1,2,\ldots$, we can consider only those players who are going to update their actions. Note that the above definition is slightly different from that given in \cite{C9}, in which the process of selecting players is done at random. In this paper, we allow arbitrary selection of players, which also includes the random selection of players. Next, we introduce the concept of upper semicontinuity \cite{C9}, which will be useful for our subsequent analysis.  

\begin{definition}{\normalfont (Upper Semicontinuous Correspondence)}\label{def:2}
A correspondence $g: X\rightrightarrows Y$ is
upper semicontinuous at $\hat{x}$, if for any open neighborhood $V$ of $g(\hat{x})$ there exists a neighborhood $U$ of $\hat{x}$ such that $g(x)\subset V$ for all $x\in U$. We say $g$ is upper semicontinuous, if it is
upper semicontinuous at every point in its domain and $g(x)$ is a compact set for each $x\in X$.
Alternatively, when $Y$ is compact, $g$ is upper semicontinuous if its graph $\{(x,y) | x\in X, y\in g(x)\}$ is a closed set.
\end{definition}

In the following we describe a tighter characterization for $D'$ in the statement of Theorem \ref{Th:1} in terms of $\epsilon$-NE set for the special case of better/best response dynamics.

\begin{prop}\label{prop:2}
In a finite discrete action space, every path of play either converges to a single point or enters a set of closed paths (cycles).
\end{prop}
\begin{proof}
Any path of play (a sequence of action profiles obtained by repeatedly applying map $K$ on itself) in a finite discrete-action space either enters a set of cycles or converges to a fixed point. This is true because, when applying map $K$ to itself, the path of play must intersect itself at some point. Thus, if the path never converges to a fixed point, it can be repeated; in this case, a collection of cycles must emerge in the action space. Any point in the action space that is not on one of these cycles is only chosen a finite number of times. Otherwise, a closed path between two successive selections of that point exists. Therefore, that point must belong to a cycle among the convergent set of cycles. This completes the proof.
\end{proof}


\begin{thm}\label{Th:2}
Let $\Gamma$ and $\Gamma^M$ denote the original and perturbed games as defined in Section \ref{Se_Nom}. Moreover, assume that $\Gamma$ is a potential game and $d(\Gamma,\Gamma^M)\leq \delta$ for some $\delta>0$. Let $K:\mathcal{A}\rightarrow \mathcal{A}$ be a better response map for the perturbed game $\Gamma^M$, where $\mathcal{A}$ denotes players' discrete-action finite set. Then, for initial actions $a^0\in \mathcal{A}$, the sequence $K^n(a^0)$ converges either to a \emph{NE} of $\Gamma^M$ or a cycle inside the $\delta|\mathcal{A}|$-\emph{NE} set $\mathcal{X}_{\delta|\mathcal{A}|}$.
\end{thm}
\begin{proof}
As Proposition~\ref{prop:2} indicates, every path of play either converges to a single point or to a set of cycles. If it converges to a fixed point, there is nothing to prove as it must be a NE of the game $\Gamma^m$ given the definition of the better response map $K$. Therefore, we only focus on a path of play that enters a  set of cycles, i.e., there are $l_0$ and $\pi$ such that $\big(K^{l_0}(a^0),K^{l_0+1}(a^0),...,K^{l_0+\pi-1}(a^0)\big)$ forms a closed path with $K^{l_0+\pi}(a^0)=K^{l_0}(a^0)$. Since for every $n$, $K^n(a^0)\in\mathcal{A}$, we can represent it by a sequence of action profiles $K^{l_0+t}(a^0)=a^{t+1}, t\in\{0,1,...,\pi-1\}$. Let $u_{n_{r}}^M$ indicate the utility of agent $n_{r}$ who is selected at step $r$ to improve its utility by moving from $a^{r-1}$ to $a^{r}$. Then, $\alpha_r:=u_{n_{r}}^M(a^r)-u_{n_{r}}^M(a^{r-1})>0$, and we can write 
\begin{equation}\label{eq:5}
\sum_{r=1}^{\pi}\big(u_{n_{r}}^M(a^r)-u_{n_{r}}^M(a^{r-1})\big)=\sum_{r=1}^{\pi}\alpha_r>\alpha_r, r\in\{1,2,...,\pi\}.
\end{equation}
Additionally, we know that in the corresponding original game $\Gamma$, which is assumed to be a potential game, we have
\begin{equation}\label{eq:6}
\sum_{r=1}^{\pi}\big(u_{n_{r}}(a^r)-u_{n_{r}}(a^{r-1})\big)=\sum_{r=1}^{\pi}\big(\phi(a^r)-\phi(a^{r-1})\big)=0.
\end{equation}
Subtracting (\ref{eq:6}) from (\ref{eq:5}), for any $r\in\{1,...,\pi\}$ we have
\begin{align}\label{pi_delta_eq}
 \pi \delta\geq \sum_{r=1}^{\pi}\Big(\left(u_{n_{r}}^M(a^r)-u_{n_{r}}^M(a^{r-1})\right)-\left(\phi(a^r)-\phi(a^{r-1})\right)\Big) \nonumber \\
 >\alpha_r .
\end{align}
Since players in the better/best response dynamics are chosen arbitrarily among the eligible set of players who can improve their utility, for every $a^r$ from this cycle the dynamics can potentially lead to another cycle such that $a^{r+1}$ with $\alpha_{r+1}=\max_{i,q_i}u_{i}^M(q_i,a^r_{-i})-u_{i}^M(a^{r})$ for some $\pi$. This means $a^r$ is a $\pi\delta$-NE of the game $\Gamma^M$. Finally, since $\pi\leq|\mathcal{A}|$, for all $r\in \{1,2,...,\pi\}$, $a^r\in \mathcal{X}_{\pi\delta}\subseteq \mathcal{X}_{|\mathcal{A}|\delta}$, which means the path of play converges to the set of $|\mathcal{A}|\delta$-NE of the game $\Gamma^M$. This completes the proof.
\end{proof}



\begin{rem}
The result of Theorem \ref{Th:2} is valid if $|\mathcal{A}|<\infty$. That is the main reason that in this section, we restricted our attention to games with discrete finite action spaces.
\end{rem}

 Using Theorem \ref{Th:2}, we conclude that regardless of choosing arbitrarily large action space $\mathcal{A}$, if there are finitely many pure equilibria for the games $\Gamma^M$ and $\Gamma$, the sequences $K^n(x^0)$ and $Z^n(x^0)$ will eventually converge to some bounded domain $D'\subset \mathbb{R}^N$. This enables us to apply Theorem \ref{Th:1} on mappings $K$ and $Z$ and compare their long-term behavior. 

In fact, we can compare our results with those in \cite{C9} for the special better/best response dynamics considered in this sub section and say this result is similar to Theorem 3.1 in \cite{C9} but with a weaker assumption yields to a more general setting. In other words, the result of Theorem \ref{Th:2} cover that of Theorem 3.1 in \cite{C9} when the probability of selecting players to response to other players action is fixed with time.

It was shown in \cite{C9} that the better/best response dynamics in both the perturbed and original games will eventually converge to $\delta |\mathcal{A}|$-NE. However, that work only provides information on the difference between those dynamics outcomes in the perturbed and the original games in terms of $\epsilon$- approximate Nash equilibrium. By analyzing the map of better/best response dynamics, in Theorem \ref{Th:1} with the aid of supplemental subsection \ref{disc_act_sp_seq_ga} in this work, in fact we stepped forward and we can quantify the distance between the convergence properties of those algorithms in the corresponding games. Therefore we would obtain tighter characterization of the limit behavior of better/best response dynamics in discrete action space sequential games. 

In the following utilizing Theorem \ref{Th:1} as a foundation with the aid of subsequent subsections as supplements, we provide an analysis for better/best response dynamics in repeated games which is unprecedented. In fact, having a look at a number of techniques from \cite{C9} for analysis of fictitious play dynamics and developing them to apply on better/best response dynamics in repeated games, we would be able to obtain an appropriate candidate of set $D'$ which would be useful in applying the Theorem \ref{Th:1}.

\subsection{Simultaneous Update Dynamics in Repeated Games with Discrete or Continuous Action Space} \label{Sub_sec_rep_game}

In this section, we consider an extension of our previous results to study the asymptotic behavior of simultaneous update dynamics in repeated games. More precisely, we assume that at each iteration of the dynamics, all the players simultaneously respond to the action profile at the previous time step. In particular, we are interested in comparing the asymptotic behavior of such simultaneous dynamics in the original and perturbed games. As before, we denote an action profile of players by $a=(a_1,\ldots,a_n)$, and a sequence of action profiles generated by the simultaneous dynamics by $\{a^t, t=1,2,\ldots\}$.



In order to continue with best response dynamics in repeated games let us define the following function. We assume for now that the action space is discrete and will provide insights into the continuous action space case toward the end of the section.
\begin{equation}\label{25}
U(f)=\sum_{a\in \mathcal{A}}\phi(a)f_1(a_1)...f_N(a_N),
\end{equation}
where $f=(f_1,f_2,...,f_N)$ and $f_j(.):\mathcal{A}_j\rightarrow [0,1]$. This function $U(f)$ has several properties such as continuity and differentiability with respect to $f$. Let us define a sequence $\{f^t\}_{t=1}^{\infty}$ such that 
\begin{align}
f_j^t=\begin{cases} 1 & \mbox{if} \ a_j=x^t_j\\
0 & \mbox{otherwise}.
\end{cases} 
\end{align}
As a result, one may write $U(f^t)=\phi(x^t)$. Defining 
\begin{align}\nonumber
&U(a_j,f_{-j})=\cr 
&\sum_{a_{-j}\in \mathcal{A}_{-j}} \phi(a_{j},a_{-j})f_1(a_1)...f_{j-1}(a_{j-1})f_{j+1}(a_{j+1})...f_{N}(a_{N}) 
\end{align}
we can write Taylor series expansion for $U(f^{t+1})$ around $f^{t}$ as follows 
\begin{align}
&U(f^{t+1})=U(f^{t})+\cr 
&\sum_{m=1}^N\sum_{a_{m}\in \mathcal{A}_{m}}\big(f^{t+1}_m(a_m)-f^{t}_m(a_m)\big)U(a_m,f_{-m})+\cr
&O(\|f^{t+1}-f^{t}\|^2).
\end{align}
Clearly, based on (\ref{25}) we have
\begin{equation}
\sum_{a_{m}\in \mathcal{A}_{m}}f^{t+1}_m(a_m)U(a_m,f^t_{-m})=U(f^{t+1}_{m},f^t_{-m}),
\end{equation}
and 
\begin{equation}
\sum_{a_{m}\in \mathcal{A}_{m}}f^{t}_m(a_m)U(a_m,f^t_{-m})=U(f^{t}_{m},f^t_{-m}).
\end{equation}
In a repeated game where every agent respond to the previous state, they will logically adopt the best response strategy if they are rational. Therefore, we focus on the best response dynamics in repeated games. The following lemma provides a lower bound on the amount by which the potential function of the potential game which is close to the original game improves for moving from $t$-th element to $(t+1)$-th element of sequence $\{x^{t}\}_{t=1}^\infty$.
\begin{lemma}\label{lem6}
Suppose that $K$ is the map of best response dynamics in a repeated game as defined in (\ref{Map_K}), and $\phi$ is the potential function of a close potential game to the underlying game to which map $K$ is associated. Assume that the $MPD$ of these two games is some $\delta\geq 0$. If $x^t$, which is defined in (\ref{Map_K}), is outside of an $\epsilon$-equilibrium, we have
\begin{equation}
\phi(x^{t+1})-\phi(x^{t})\geq \epsilon -N\delta+O(\|f^{t+1}-f^{t}\|^2). 
\end{equation}
\end{lemma}
\begin{proof}
We can write 
\begin{align}\label{26}
\phi(x^{t+1})=\phi(x^{t})+\sum_{m=1}^N \phi(x^{t+1}_{m},x^t_{-m})-\phi(x^{t}_{m},x^t_{-m})&\nonumber\\
+O(\|f^{t+1}-f^{t}\|^2).&
\end{align}
Since by the definition of $MPD$ for every $m\in \mathcal{M}$ we have 
\begin{align}
u^M_m(x^{t+1}_{m},x^t_{-m})-u^M_m(x^{t}_{m},x^t_{-m})-\big[ \phi(x^{t+1}_{m},x^t_{-m})- & \nonumber \\ \phi(x^{t}_{m},x^t_{-m})\big]
\leq \delta,&
\end{align}
we can write 
\begin{align}\label{27}
&\phi(x^{t+1})\geq \phi(x^{t})+\cr
&\sum_{m=1}^N \big[u^M_m(x^{t+1}_{m},x^t_{-m})-u^M_m(x^{t}_{m},x^t_{-m})-\delta\big]+\cr
& O(\|f^{t+1}-f^{t}\|^2).
\end{align}
By the assumption of the lemma for every $m\in\{1,2,...,N\}$,
\begin{equation}
u^M_m(x^{t+1}_{m},x^t_{-m})-u^M_m(x^{t}_{m},x^t_{-m})\geq 0.
\end{equation}
Moreover, since $x^t$ is outside of an $\epsilon$-equilibrium, for some $m_0$ we can write

\begin{equation}
u^M_{m_0}(x^{t+1}_{m_0},x^t_{-m_0})-u^M_{m_0}(x^{t}_{m_0},x^t_{-m_0})\geq \epsilon.
\end{equation}
Thus,
\begin{equation}\label{28}
\phi(x^{t+1})\geq \phi(x^{t})+\epsilon-N\delta+O(\|f^{t+1}-f^{t}\|^2).
\end{equation}
This completes the proof.
\end{proof}
In Lemma \ref{lem6}, $f^{t+1}$ and $f^{t}$ are vectors of length $n$ of which elements belong to $\{0,1\}$. Let us introduce $\{k^t\}_{t=1}^{\infty}$ and $\{w^t\}_{t=1}^{\infty}$ such that for every $t$ we have $k^t=O(\|f^{t+1}-f^{t}\|^2)$ and $w^t=\|x^{t+1}-x^{t}\|$. By construction $\{k^t\}_{t=1}^{\infty}$ is a bounded sequence. Assuming that we can guarantee $\{w^t\}_{t=1}^{\infty}$ is a bounded sequence as well. We can then claim that $\limsup_{t}k^t$ and $\limsup_{t}w^t$ exist. The following assumption includes some conditions required for the rest of our analysis to be able to give some results for best response dynamics in both of discrete and continuous action spaces.

\begin{assum}\label{Cont_Int_Util}
Functions $u_i(a), u^M_i(a), i\in \mathcal{M}$ that are defined over action space $\mathcal{A}$ in \eqref{eq_1} are Lipschitz continuous over $Conv(\mathcal{A})$.
\end{assum}

The following lemma and the ensuing theorem provide a more precise characterization of the limiting behaviour of the dynamics associated with a repeated game in the form introduced earlier.

\begin{lemma}\label{lem7}
Let $\alpha\geq 0$, $\theta\geq 0$ be given and Assumption~\ref{Cont_Int_Util} holds. There is some $L$ such that if $\|x-y\|\leq \theta$ and $x\in \mathcal{X}_\alpha$ then $y\in \mathcal{X}_{\alpha+L\theta}$.
\end{lemma}
\begin{proof}
Let $\nu:Conv(\mathcal{A})\rightarrow \mathbb{R}$ be a function such that 
\begin{equation}
\nu(x)=-\max_{m\in \mathcal{M}, p_m\in \mathcal{A}_m}\big(u_m(p_m,x_{-m})-u_m(x_m,x_{-m})\big).
\end{equation}
It follows from the definition of $\epsilon$-equilibrium that a profile $x$ is an $\epsilon$-equilibrium if and only if $\nu(x)\geq -\epsilon$.
Since $u_m(p_m,x_{-m})$ and $u_m(x_m,x_{-m})$ are Lipschitz continuous in $x$, the term $u_m(p_m,x_{-m})-u_m(x_m,x_{-m})$ is also Lipschitz continuous in $x$. Let us write $h_m(d,x)=u_m(d,x_{-m})-u_m(x_m,x_{-m})$ for $d\in \mathcal{A}_m$ which is closed and bounded. Firstly, one observes that since $h_m(d,x)$ is Lipschitz continuous in $d$ and $\mathcal{A}_m$ is compact,
\begin{equation}
\sup_{d\in \mathcal{A}_m}h_m(d,x)=\max_{d\in \mathcal{A}_m}h_m(d,x).
\end{equation}

Secondly, define $L(.)$ as the Lipschitz constant of the function inside the argument and consequently $L(h_m(d,x))$ as the Lipschitz constant of $h_m(d,x)$. Utilizing the definition of Lipschitz continuity, we can write $L(h_m(d,x))\leq 2L_{u_m}$ where $L_{u_m}$ is the Lipschitz constant of $u_m(.)$. Thirdly, According to Proposition 1.32 of \cite{C27}, since $h_m(d,x)$ is not infinite anywhere on its domain,  for family $\{h_m(d,x)\}_d$ of Lipschitz continuous functions we can claim that $\sup_{d}h_m(d,x)$ is a Lipschitz function with 
\begin{equation}
L\big(\sup_{d}h_m(d,x)\big)\leq \sup_{d}L\big(h_m(d,x)\big).
\end{equation}
Considering all these three points together one can write
\begin{equation}
L\big(\max_{d\in \mathcal{A}_m}h_m(d,x)\big)\leq 2L_{u_m}.
\end{equation}
 Therefore, with the same argument we can claim $\nu(x)$ is Lipschitz continuous with Lipschitz constant less than or equal to $L^0=\max_{m\in \mathcal{M}}2L_{u_m}$. It follows that if $\|x-y\|\leq\theta$, then $\|\nu(x)-\nu(y)\|\leq L^0\theta$. Thus,
 \begin{equation}
-L^0\theta\leq \nu(y)-\nu(x)\leq L^0\theta.
 \end{equation}
 Since $x\in \mathcal{X}_\alpha$ we have $\nu(x)\geq -\alpha$. As a result, $\nu(y)\geq -L^0\theta-\alpha$, and consequently, $y\in \mathcal{X}_{L^0\theta+\alpha}$. This completes the proof.
\end{proof}

We note that using Assumption \ref{Cont_Int_Util}, since $\mathcal{A} \subset Conv(\mathcal{A})$, Lemma \ref{lem7} is valid for both discrete and continuous action spaces. 

In the following theorem, we introduce an invariant set which describes the asymptotic behavior of the best response dynamics in repeated games.
\begin{thm}\label{th:4}
Let $T_0>0$ be an arbitrary positive integer and Assumption~\ref{Cont_Int_Util} holds true. The sequence $\{x^t\}_{t=1}^{\infty}$ indicating the dynamics of a repeated game, that is generated by compositions of map $K$ on itself as defined in (\ref{Map_K}), converges to the set $C_{T_0,\epsilon}$ for any $\epsilon>0$. i.e. $\exists T_\epsilon>0$ such that for all $t \geq T_\epsilon \geq T_0$, 
denoting 
\begin{equation}
R_4=N\delta+sup_{t\geq T_0}k^t+\epsilon,
\end{equation}
and 
\begin{equation}
R_5=N\delta+sup_{t\geq T_0}k^t+L^0 sup_{t\geq T_0}w^t+\epsilon,
\end{equation}
where
\begin{equation}
 C_{T_0,\epsilon}=\{x\in Conv(\mathcal{A})| \phi(x)\geq min_{y\in \mathcal{X}_{R_5}}\phi(y)\}. 
\end{equation}
\end{thm}
\begin{proof}
Based on the construction of $\mathcal{X}_{R_4}, \mathcal{X}_{R_5},$ and $ C_{T_0,\epsilon}$
\begin{equation}
\mathcal{X}_{R_4}\subseteq \mathcal{X}_{R_5}\subseteq C_{T_0,\epsilon}.
\end{equation}
We firstly show that $\mathcal{X}_{R_4}$ will be visited infinitely often. Consider ($N\delta+sup_{t\geq T_0}k^t+\epsilon$)-equilibrium set. Suppose that $\exists T_1 \geq T_0$ such that $\forall t\geq T_1$ we have $x^t$ is out of ($N\delta+sup_{t\geq T_0}k^t+\epsilon$)- equilibrium set. Thus, using lemma \ref{lem6}, 
\begin{equation}\label{eq:29}
\phi(x^{t+1})-\phi(x^{t})\geq N\delta+sup_{t\geq T_0}k^t+\epsilon -\big(N\delta+k^t\big)\geq \epsilon.
\end{equation}
Hence,
\begin{equation}\label{eq:30}
\limsup_{t}\phi(x^{t+1})-\phi(x^{T_1})\geq \sum_{t\geq T_1}\epsilon.
\end{equation}
The left hand side of (\ref{eq:30}) is bounded while the right hand side of it is unbounded. This yields a contradiction which means the ($N\delta+sup_{t\geq T_0}k^t+\epsilon$)-equilibrium ($\mathcal{X}_{R_4}$) will be visited infinitely often. Whenever $x^t$ visit this set since $\|x^{t+1}-x^{t}\|=w^t\leq sup_{t\geq T_0}w^t$, according Lemma \ref{lem7}, $x^{t+1}$ belongs to $\mathcal{X}_{R_5}$. Thus, if $x^t$ visits ($N\delta+sup_{t\geq T_0}k^t+\epsilon$)-equilibrium set, then $x^{t+1}$ visits ($N\delta+sup_{t\geq T_0}k^t+L^0 sup_{t\geq T_0}w^t+\epsilon$)-equilibrium set. For some $T_1\geq T_0$, we have $x^{T_1}\in C_{T_0,\epsilon}$. (we know it occurs since $\{x^t\}_{t=1}^{\infty}$ visit $\mathcal{X}_{R_4}\subseteq C_{T_0,\epsilon}$ infinitely many times.) Then two cases may happen.\\
\textbf{Case1.} If $x^{T_1}\in \mathcal{X}_{R_4}$, we just proved that 
\begin{equation}
x^{T_1+1}\in \mathcal{X}_{R_5}\subseteq C_{T_0,\epsilon}.
\end{equation}
\textbf{Case2.} If $x^{T_1}\in C_{T_0,\epsilon}- \mathcal{X}_{R_4}$, since $x^{T_1}$ is outside $\mathcal{X}_{R_4}$, according to (\ref{eq:29}) we have $\phi(x^{T_1+1})\geq \phi(x^{T_1})$. Additionally, since $x^{T_1}\in C_{T_0,\epsilon}$, we have
\begin{equation}
\phi(x^{T_1})\geq min_{y\in \mathcal{X}_{R_5}}\phi(y).
\end{equation}
This means that $x^{T_1+1} \in C_{T_0,\epsilon}$. Therefore, from $T_1$ on, every elements of $\{x^t\}_{t=1}^{\infty}$ belongs to $C_{T_0,\epsilon}$. As a result, for every $T_0$, the path of play will be eventually trapped in $C_{T_0,\epsilon}$. This completes the proof.
\end{proof}
Denoting 
\begin{equation}
R_6= N\delta+\limsup_{t}k^t+L^0 \limsup_{t}w^t+\epsilon,
\end{equation}
since the argument of Theorem \ref{th:4} is valid for every $T_0$, path of play will eventually be trapped in $\{\inf_{T_0}C_{T_0,\epsilon}\}=C_{\epsilon}$, where
\begin{equation}
C_{\epsilon}=\{x\in Conv(\mathcal{A})| \phi(x)\geq min_{y\in \mathcal{X}_{R_6}}\phi(y)\}.
\end{equation}
Taking a similar proof to that of Corollary 5.1 of \cite{C9}, one can show that the path of play converges to $C_0$.
\begin{rem}
For any repeated game algorithms satisfying the conditions in the statement of Theorem \ref{th:4} for which $\{x^t\}_{i=1}^{\infty}$ converges, we have $\limsup_{t}w^t=0$. Therefore, set $C_0$ boils down to
\begin{equation}
C_{0}=\{x\in Conv(\mathcal{A})| \phi(x)\geq min_{y\in \mathcal{X}_{N\delta+\limsup_{t}k^t}}\phi(y)\}.
\end{equation}

Moreover, for discrete-action spaces, as $\{x^t\}_{t=1}^{\infty}$ converges, $\{k^t\}_{t=1}^{\infty}$ converges to $0$, so does $\limsup_{t}w^t$. Thus, for this case we can write 
\begin{equation}
C_{0}=\{x\in Conv(\mathcal{A})| \phi(x)\geq min_{y\in \mathcal{X}_{N\delta}}\phi(y)\}.
\end{equation}
\end{rem}
We note that the results of this subsection (Subsection \ref{Sub_sec_rep_game}) is valid for both discrete and continuous action spaces. Proofs for best response dynamics in games with continuous action spaces follows by replacing $\int$ with $\sum$ in all of the equations from (\ref{25}) to the end of the subsection.

\subsection{Response to an Estimation of Past Actions} \label{Res_Est_past_act}
In this subsection,  we provide some results regarding the case where each player responds to estimates of other players action, which are eventually accurate, instead of the actual actions. The significance of this problem is evident in case of network games where agents have no information over some other agents' actions, which are not their neighbors, and only some estimation of them are available for the agents. We claim that regardless of the underlying learning process as long as this estimation becomes arbitrarily accurate after passing sufficient time, the asymptotic behavior converges to a quantifiable set. Let us assume that maps $K$ and $Z$ are defined as in Section~\ref{Se_Nom} using \eqref{Map_K} and \eqref{Map_Z}, respectively.
To model this, let us first formulate the estimates of the players' actions by other players.
We first introduce a sequence $\{\hat{K}_n(x^0)\}_{n=1}^\infty$ such that for a given $n$, $\hat{K}_n(x^0)$ is an estimation of $K^n(x^0)$ and $\hat{K}_n(x^0)\in \mathbb{R}^N$. Then the following assumption is made to guarantee that the estimation becomes accurate asymptotically.
\begin{assum}\label{as_accu}
The estimation of $K^n(x_0)$ which is $\hat{K}_n(x^0)$ becomes more accurate as time grows, i.e., $\|\hat{K}_n(x^0)-K^n(x^0)\|\rightarrow 0$ as $n\rightarrow\infty$.
\end{assum}

Clearly, if part \eqref{Th:1_1} in Theorem \ref{Th:1} holds, we can write
\begin{equation}
\|K(x)-K(y)\|\leq 2(\delta_1+\delta_2)+L_C\|x-y\|,
\end{equation}
for all $x,y\in D'$ and $\delta_1,\delta_2,$ and $L_C$ defined in the statement of Theorem \ref{Th:1}. According to part 1 of Theorem~\ref{Th:1}, $K^n(x^0)$ will eventually be trapped in a compact set $S \subseteq D'$. Considering Assumption \ref{as_accu}, it is not difficult to show that $K\big(\hat{K}_n(x^0)\big)$,  which denotes the actions that players take at step $n+1$ with respect to the estimates of actions at step $n$, will eventually be trapped in the compact set 
\begin{equation}
S_d=\{x|\min_{y\in S}\|x-y\|\leq 2(\delta_1+\delta_2)\}.
\end{equation}
Moreover, we can write 
\begin{align}
\|K\big(\hat{K}_n(x^0)\big)-C^{n+1}(x^0)\|&\leq \|K\big(\hat{K}_n(x^0)\big)-K^{n+1}(x^0)\|\cr 
&+\|K^{n+1}(x^0)-C^{n+1}(x^0)\|.
\end{align}
This means as $n \rightarrow \infty$, 
\begin{equation}
\|K\big(\hat{K}_n(x^0)\big)-C^{n+1}(x^0)\|\leq (\delta_1+\delta_2)\frac{3-2L_C}{1-L_C}.
\end{equation}
Thus, if the players respond to a asymptotically accurate estimation of actions of other players players' actions will be eventually  trapped in the intersection of $N_{(\delta_1+\delta_2)(3-2L_C)/(1-L_C)}[x^*]$ and the set $S_d$.

\section{Application to Cournot Duopoly Games}\label{Example_and_Simulation}

This section demonstrates how the results developed in the previous sections can be applied to an example from Cournot duopoly games. In this form of the Cournot game, two firms compete with each other, where firm $i\in\{1,2\}$ produces $a_i$ units of good, at the marginal cost of $ c_i\geq 0$, and sells it at a price of 
\begin{equation}
P=\max\{d-Q,0\},
\end{equation}
where 
\begin{equation}
Q=a_1+a_2
\end{equation}
is the total supply,  and $d\geq c_i$. Each firm wants to maximize its profit through maximizing the following payoff function:
\begin{equation}
u_i=a_i(P-c_i).
\end{equation}
The best response dynamics for agent $i$ in this class of games can be implemented as follows. If $a_{j}>d, j\neq i$, then $P=0$, and  firm $i$ can at best make zero profit. Otherwise, if $a_{j}\leq d, j\neq i$, for any $a_i\in (0, d-a_{j}), j\neq i$, the utility derived by firm $i$ is given by:
\begin{equation}\label{utility_cournot}
u_i(a_i,a_{j})=a_i(d-a_i-a_{j}-c_i).
\end{equation}
Taking the derivative of (\ref{utility_cournot}) and setting it equal to zero, the best response of firm $i$, denoted by $a_i^B$, is given by 
\begin{equation}
a_1^B(a_2)=\max\{0, \frac{d-a_2-c_1}{2}\},
\end{equation}
\begin{equation}
a_2^B(a_1)=\max\{0, \frac{d-a_1-c_2}{2}\}.
\end{equation}
In a repeated Cournot duopoly game, by starting from an arbitrary initial action $(a_1^0,a_2^0)$, and considering $a_i \geq 0$, we obtain the following best response dynamics: 
\begin{equation}\label{eq-cor-1}
 a_1^{t+1}(a_2^t)=max\{0,\frac{d-a_2^t-c_1}{2}\},
\end{equation}
\begin{equation}\label{eq-cor-2}
 a_2^{t+1}(a_1^t)=max\{0,\frac{d-a_1^t-c_2}{2}\}.
\end{equation}
Therefore, in this particular example, $Z(a_1^t,a_2^t)=(a_1^{t+1},a_2^{t+1})$ for $a_1^t,a_2^t,a_1^{t+1}$ and $a_2^{t+1}$ as defined by \eqref{eq-cor-1} and \eqref{eq-cor-2}.
In the following, we first investigate the form of invariant set $C_0$ for the class of Cournot games in the case of a duopoly. Then we apply Theorem \ref{Th:1} to find the asymptotic behavior of the best response dynamics approximately and compare it with the existing accurate analysis for this class of games.
It is known that Cournot duopoly games with utility functions of the form of (\ref{utility_cournot})
are potential games with the potential function \cite{C29}: 
\begin{equation}\label{potential_cournot_duopoly}
    \phi(a)=d(a_1+a_2)-(a_1^2+a_2^2)-(a_1 a_2)-c_1a_1-c_2a_2.
\end{equation}The maximum of this potential function takes place at $(a_1,a_2)=\big((d-c)/3,(d-c)/3\big)$, with the maximum value of $(d-c)^2/3$, while the minimum of the potential function for the rectangular action space $0\leq a_1,a_2 \leq \bar{a}$ takes place at $(\bar{a},\bar{a})$, with the minimum value of $2(d-c)\bar{a}-3\bar{a}^2$. For the duopoly Cournot game, a bound for the term $O(\|f^{t+1}-f^t\|)$ can be obtained based on calculating the remainder of Taylor series expansion theorem \cite{C28}. For a bi-variable case of the first order, an upper bound for the absolute value of the remainder would be obtained by 
\begin{equation}
\frac{M}{2!}\|h\|^2,
\end{equation}
where $h=f^{t+1}-f^t$, and $M$ is an upper bound for $\|\partial^2 U(f)/(\partial f^1 \partial f^2)\|$. For a rectangular continuous action space where $0\leq a_1 \leq \bar{a}_1$ and $0\leq a_2 \leq \bar{a}_2$, we have 
\begin{equation}
M=|\int_0^{\bar{a}_1} \int_0^{\bar{a}_2} \phi(a) d a_2 d a_1|.
\end{equation}
Substituting \eqref{potential_cournot_duopoly} into the above expression, we obtain
\begin{equation}
M=\Big|(d-c)\big(\frac{\bar{a}_1^2\bar{a}_2+\bar{a}_1\bar{a}_2^2}{2}\big)-\big(\frac{\bar{a}_1^3\bar{a}_2+\bar{a}_1\bar{a}_2^3}{3}+\frac{\bar{a}_1^2\bar{a}_2^2}{4}\big)\Big|
\end{equation}
In particular, for the special case of $\bar{a}_1=\bar{a}_2=\bar{a}$, we get
\begin{equation}
M=|(d-c)\bar{a}^3-\frac{11}{12}\bar{a}^4|.
\end{equation}It is not difficult to show that the map $Z$ is contractive. As a result, the dynamics generated by $Z$  converge, and $w^t\rightarrow 0$ as $t\rightarrow\infty$, where $w^t$ is defined as in Section \ref{Sub_sec_rep_game}. Hence, we have
\begin{equation}
    C_0=\{x\in Conv(\mathcal{A})| \phi(x)\geq min_{y\in \mathcal{X}_{2M}} \phi(y) \}.
\end{equation}

Now let us assume there is a perturbation in the utility functions of the form 
\begin{equation}\label{Perturbation_Eq}
    u^M_i(a)=u_i(a)+\Delta u_i(a),\  i\in \{1,2\},
\end{equation}
where we assume that the perturbations $\Delta u_i(a), i\in \{1,2\}$ are Lipschitz continuous with Lipschitz constants $L_{\Delta u_i}$. Thus, the dynamics of the perturbed game can be 
derived by solving the following nonlinear equations
\begin{equation}\label{81}
-a_1^{t+1}+ \frac{d-a_2^t-c_1}{2}+\frac{d\Delta u_1(a_1,a_2)}{2d a_1}|_{a_1=a_1^{t+1},a_2=a_2^{t}}=0,
\end{equation}
\begin{equation}\label{82}
-a_2^{t+1}+ \frac{d-a_1^t-c_2}{2}+\frac{d\Delta u_2(a_1,a_2)}{2d a_2}|_{a_1=a_1^{t},a_2=a_2^{t+1}}=0.
\end{equation}We note that $a_1$ or $a_2$ takes the value of zero if the solution of the corresponding equation results in a negative value. Knowing that $Z$ is a contractive map with a fixed point $\Tilde{x}$, and considering $K$ as the map of the perturbed game,  which is expressed by \eqref{81} and \eqref{82}, we can obtain set $D'$ for this sequence by using similar calculations for the map $K$. However, we skip the derivations of $D'$ here as it does not play a role in our subsequent analysis. By applying Theorem \ref{Th:1} with $\delta_1=\sqrt{L^2_{\Delta u_1}+L^2_{\Delta u_2}}$, we get $\delta_2=0$ and $L_C=L_Z=\frac{1}{2}$. That implies that the perturbed dynamics $K^n(x_0)$ converge to the set $N_{2\delta_1}(\Tilde{x})$ for any initial action $x_0$.

It is worth mentioning that the same result holds for sequential Cournot games in which players take turns to take action instead of playing simultaneously. In other words, in the original and perturbed sequential game, the dynamics would be similar as defined in \eqref{eq-cor-1}, \eqref{eq-cor-2}, \eqref{81}, and \eqref{82}, except that due to the sequential nature of the play, in each step only the dynamics of one of the players change and the other one would remain the same. As a result,  following the same procedure introduced for the repeated Cournot game, the same bound would be achieved for the sequential Cournot game.

To illustrate the benefit of the above analysis, we next consider a special numerical example for the repeated and sequential Cournot game and its perturbation by assuming $d=400$, $c_1=200,c_2=100$, and
\begin{equation}
 \Delta u_i(a)=\frac{1}{1+ e^{-((a_1-\mu_1)^2+(a_2-\mu_2)^2)^2}}.   
\end{equation}

Moreover, we set $\mu_1$ and $\mu_2$ as the coordinates of the original Nash equilibrium since we want the perturbed game to have the maximum utility difference with the Nash equilibrium of the original game. To form the dynamics, we calculate $\frac{d\Delta u_1(a_1,a_2)}{2d a_2}$ in equations \eqref{81} and \eqref{82} as
\begin{equation}
\frac{2(a_1\!-\!\mu_1)((a_1\!-\!\mu_1)^2+(a_2\!-\!\mu_2)^2)(e^{-((a_1-\mu_1)^2+(a_2-\mu_2)^2)^2})}{(1+ e^{-((a_1-\mu_1)^2+(a_2-\mu_2)^2)^2})^2},
\end{equation}
\begin{equation}
\frac{2(a_2\!-\!\mu_2)((a_1\!-\!\mu_1)^2+(a_2\!-\!\mu_2)^2)(e^{-((a_1-\mu_1)^2+(a_2-\mu_2)^2)^2})}{(1+ e^{-((a_1-\mu_1)^2+(a_2-\mu_2)^2)^2})^2}.
\end{equation}

\begin{figure}[]
    \includegraphics[]{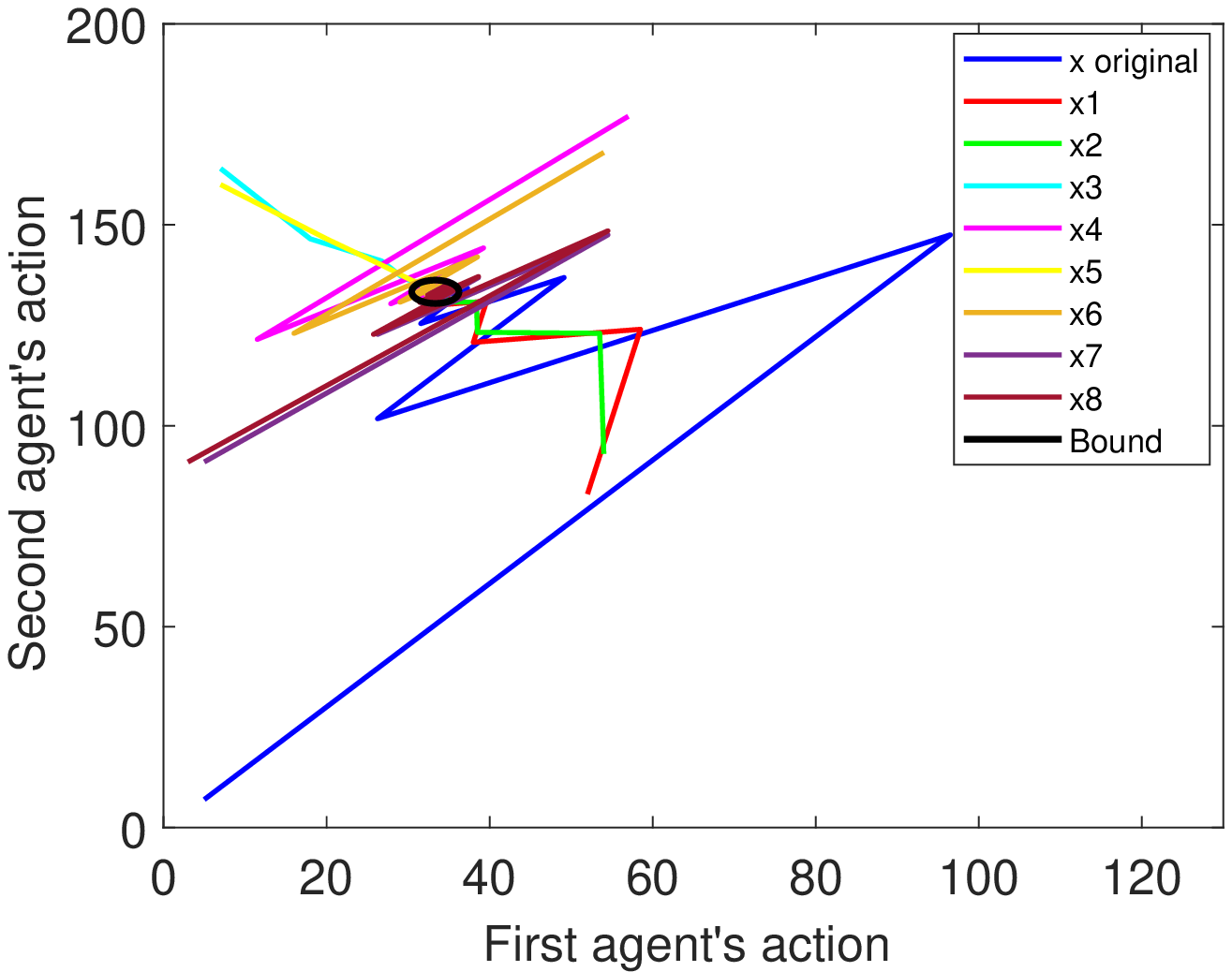}
        \includegraphics[]{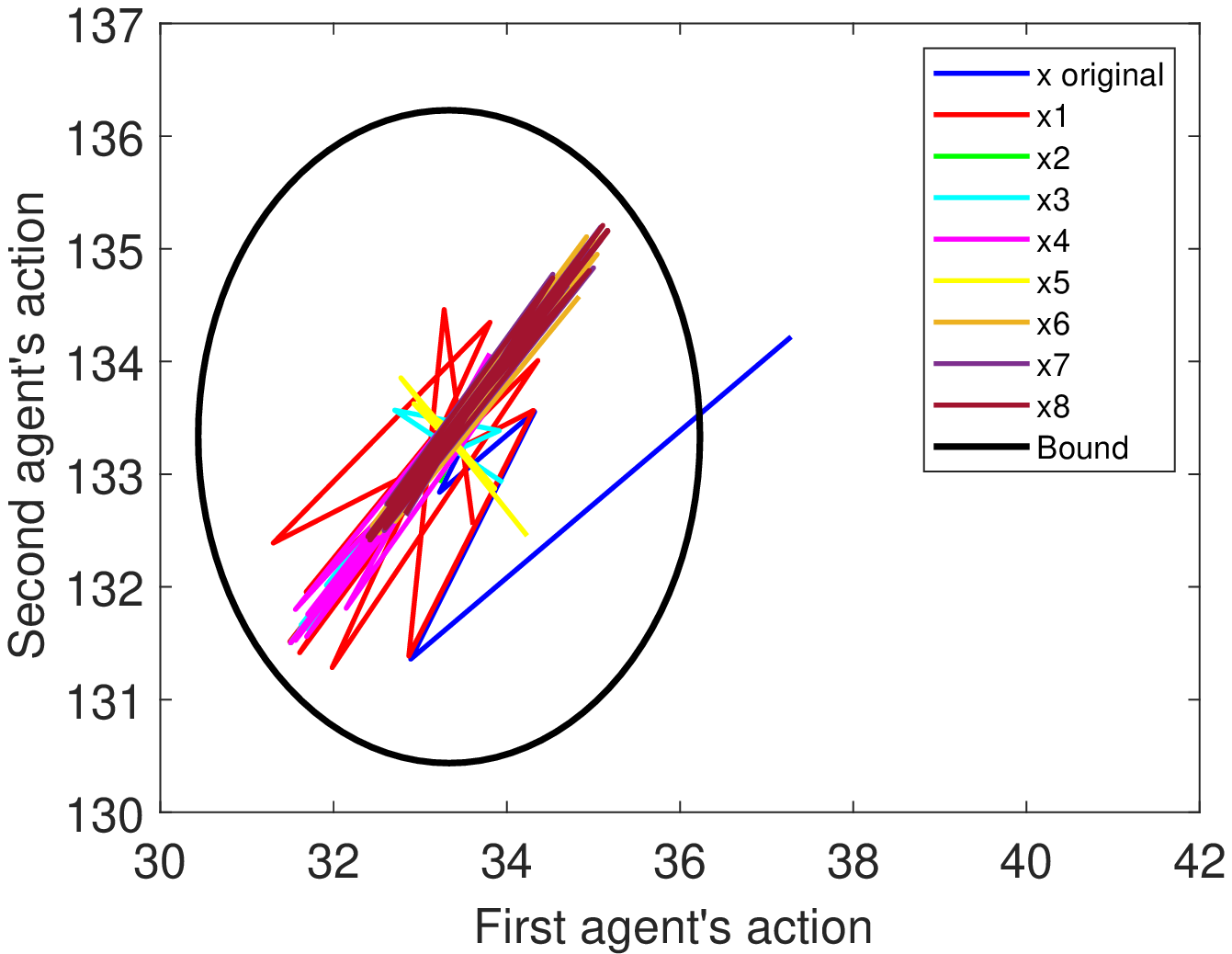}
    \centering
    \caption{The dynamics of both the original repeated game and the perturbed one with 8 different initial states}
    \label{fig:4}
\end{figure}

\begin{figure}[]
    \includegraphics[]{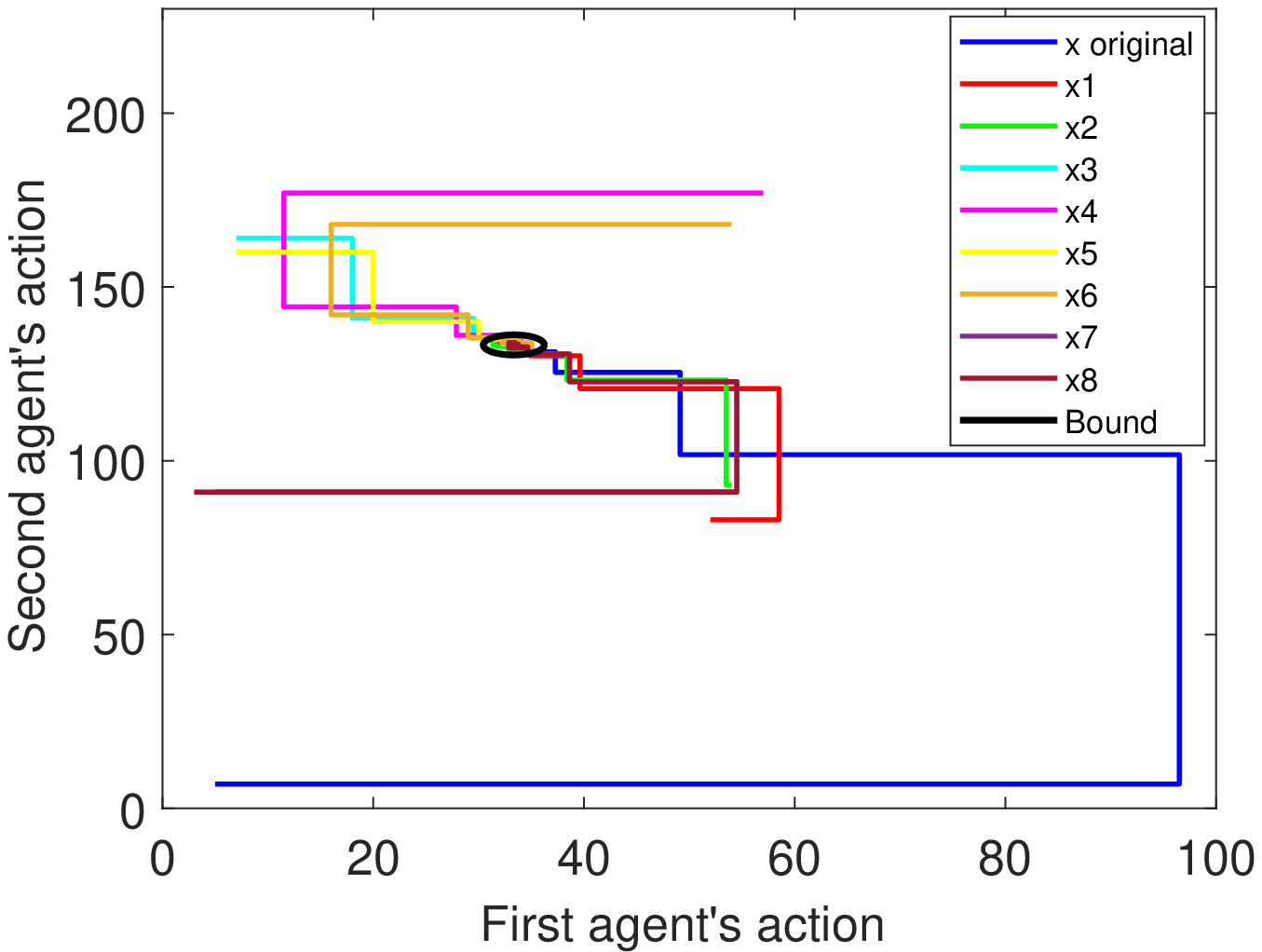}
        \includegraphics[]{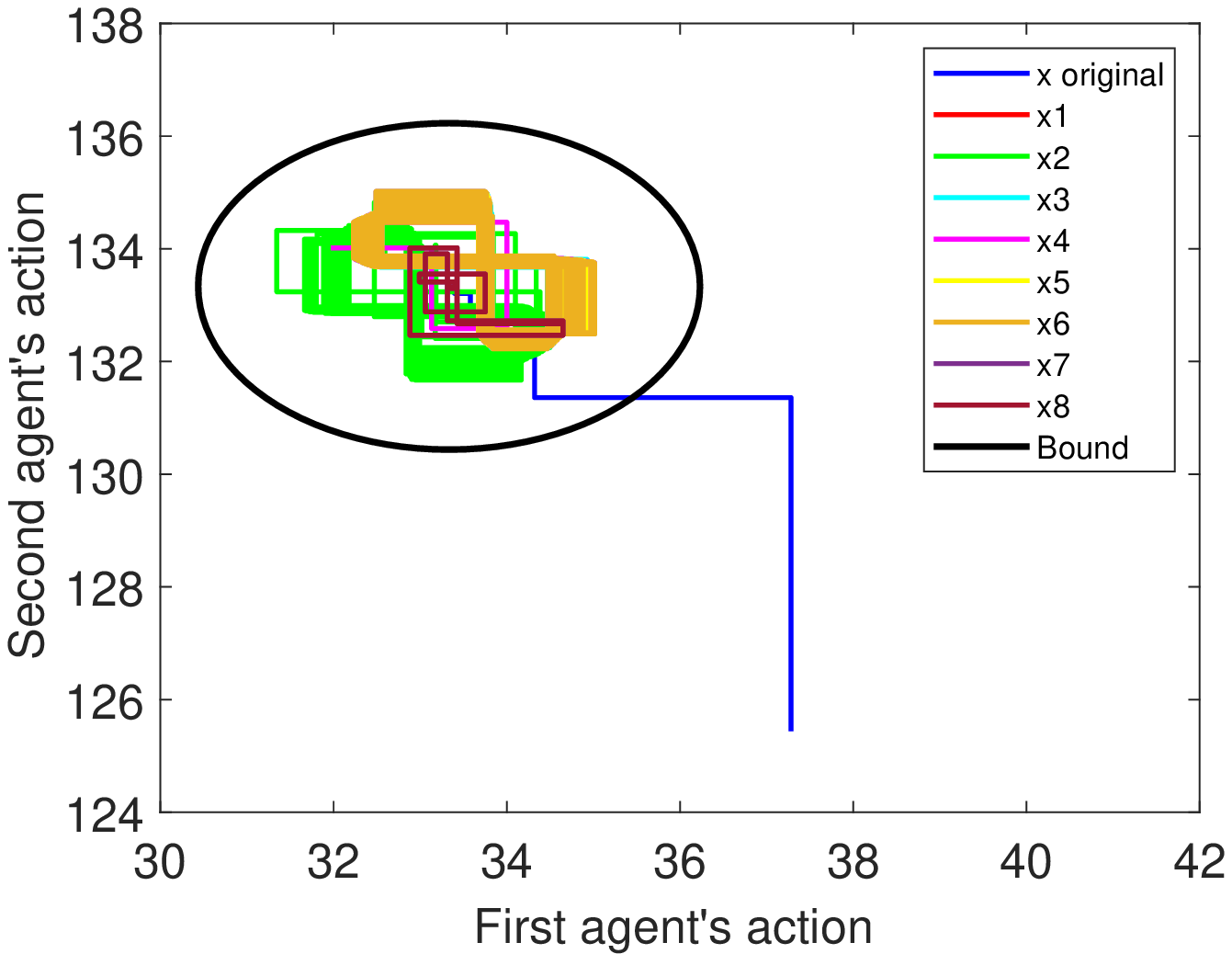}
    \centering
    \caption{The dynamics of both the original sequential game and the perturbed one with 8 different initial states}
    \label{fig:5}
\end{figure}
The dynamics of both the original repeated game and the perturbed one, and also the original sequential game and its perturbed one, with eight different initial states, are depicted in Figure~\ref{fig:4} and Figure~\ref{fig:5}, respectively. In each of these plots, the horizontal and vertical axes show the action space of each player. The bottom plot in both figures is a closer view of the same trajectories as in the top plot, from the 7th step onward, which is a good indicator of the limit behavior in these special examples. The black circle depicts our derived bound for which the limit behavior would be stuck. As it can be seen, 1) all dynamics finally get stuck in the black circle, and 2) the derived bound is a good indicator of the limit behavior and is far from being conservative, although it may not be minimal. 

\section{Conclusion}\label{Conclusion}
In this paper, we investigated a new framework to analyze the asymptotic behavior of games based on the similarity of their mappings to contractive mappings. In particular, we established some convergence results for the better/best response dynamics in the discrete-action sequential games as well as in the simultaneous finite-action repeated games. We showed that the asymptotic behavior of the dynamic game could be estimated via the fixed point of a contractive map, which is found to be close to the map of the dynamics (with some quantifiable error term). In the proposed framework, we also discussed the perturbation of the map of a given dynamics game and investigated the asymptotic behavior of the perturbed dynamic game considering the size of the perturbation. Finally, we illustrated the practicality of this framework on duopoly Cournot games.


\appendix
\section{Appendix I}\label{appendix} 

\begin{lemma}\label{lem0}
Let $\{\eta_i\}_{i=1}^\infty$ be a bounded sequence in $\mathbb{R}$, and set $\alpha=\limsup_n{\eta_n}$. Then, the following statements hold simultaneously:
\begin{enumerate}
    \item $\forall \epsilon>0$, there exist infinitely many indices $n$ such that $\alpha-\epsilon<\eta_n$.\label{lem0_1}
    \item $\forall \epsilon>0$, there are only finitely many indices $n$ such that $\alpha+\epsilon<\eta_n$.\label{lem0_2}
\end{enumerate}
\end{lemma}
\begin{proof}
For a bounded sequence $\{\eta_i\}_{i=1}^\infty$ in $\mathbb{R}$, we have $\limsup_n{\eta_n}=\max\overline{\{\eta_n\}}$. Therefore, there exists a subsequence of $\{\eta_i\}_{i=1}^\infty$  that converges to $\alpha$. As a result, for every $\epsilon>0$, elements of such subsequence will eventually lie within $N_{\epsilon}(\alpha)$, and hence, they must be all greater than $\alpha-\epsilon$. Therefore, there are infinitely many indices $n$ such that $\alpha-\epsilon<\eta_n$.

We show the second statement by contradiction. Let us assume that statement 2) does not hold. Then, there is $\epsilon_0>0$ such that infinitely many indices $n$ satisfying $\alpha+\epsilon<\eta_n$ exist. As a result, there exists a bounded subsequence of $\{\eta_i\}_{i=1}^\infty$ such that all of its elements are greater than $\alpha+\epsilon$. Therefore, $\{\eta_i\}_{i=1}^\infty$ has a convergent subsequence whose limit is greater than $\alpha+\epsilon$, which contradicts the fact that $\alpha$ is an upper bound for $\overline{\{\eta_n\}}$. This completes the proof.
\end{proof}



\begin{proof}{(Proof of  Lemma \ref{lem2}):}
Consider the space of all mappings $\mathcal{V}:\mathbb{R}^N\to\mathbb{R}^N$, which forms a vector space over $\mathbb{R}$. Consider the subset $\mathcal{S} \subseteq \mathcal{V}$ containing all the bounded Lipschitz continuous functions defined on $D'$. Clearly, this set is nonempty since any constant continuous function belongs to $\mathcal{S}$. Moreover, for every $U,V \in \mathcal{S}$ and for all $\lambda \in \mathbb{R}$, $ U+\lambda V \in \mathcal{S}$, which means that $\mathcal{S}$ forms a linear subspace of $\mathcal{V}$. Therefore, for an arbitrary component $Z\in \mathcal{V}$, the following optimization problem for some well-defined norm $\|\cdot\|$ over the space of functions has an optimal solution:
\begin{align}\label{eq:2}
& \min_{F} \|Z- F\|\nonumber \\
& s.t.~~F\in \mathcal{S}
\end{align}
Since the constraint set of this optimization problem forms a hyperplane in $\mathcal{V}$, there is always a
nonempty feasible set for this problem. The optimal solution can be obtained through the projection of $Z$ on the
hyperplane $\mathcal{S}$. In fact, this optimal solution estimates the mapping $Z$ with the closest possible Lipschitz continuous function. Let us denote the optimal solution of (\ref{eq:2}) by $F_0$. Since $F_0$ is bounded on $D'$, for every $0\leq \alpha\leq 1$, there exists some $\delta_{\alpha}>0$ such that $\|F_0-\alpha F_0\|\leq \delta_{\alpha}$. Since $\|Z-F_0\|\leq \delta$ for some $\delta$, we have $\|Z-\alpha F_0\|\leq \delta+\delta_{\alpha}$, where $\alpha F_0$ has a sufficiently small Lipschitz constant
to be a contractive mapping. Now, the function $\alpha F_0$ satisfies the conditions in Lemma \ref{lem2}. Moreover, for any other function belonging to the feasible set (\ref{eq:2}) the above method leads to another ``close" contractive mapping (although such a close contractive mapping may not be the closest one to $Z$).
\end{proof}

\bibliographystyle{plain}        

\end{document}